%% file: main.tex
\documentclass[a4paper]{llncs}

\input{defs}

\title{Improved Power Decoding of One-Point Hermitian Codes}
\author{Sven Puchinger\inst{1} \and Irene Bouw\inst{2} \and Johan Rosenkilde n\'e Nielsen\inst{3}}
\institute{Institute of Communications Engineering, Ulm University, Ulm, Germany\\
\email{sven.puchinger@uni-ulm.de}
\and
Institute of Pure Mathematics, Ulm University, Ulm, Germany\\
\email{irene.bouw@uni-ulm.de}
\and
Department of Applied Mathematics \& Computer Science, Technical University of Denmark, Lyngby, Denmark\\
\email{jsrn@jsrn.dk}}

\begin{document}

\maketitle

\begin{abstract}
We propose a new partial decoding algorithm for one-point Hermitian codes that can decode up to the same number of errors as the Guruswami--Sudan decoder.
Simulations suggest that it has a similar failure probability as the latter one.
The algorithm is based on a recent generalization of the power decoding algorithm for Reed--Solomon codes and does not require an expensive root-finding step. In addition, it promises improvements for decoding interleaved Hermitian codes.
\end{abstract}

\section{Introduction}

One-point Hermitian ($1$-H) codes are algebraic geometry codes that can be decoded beyond half the minimum distance.
Most of their decoders are conceptually similar to their Reed--Solomon (RS) code analogs, such as the \emph{Guruswami--Sudan} (GS) algorithm \cite{guruswami1998improved} and \emph{power decoding} (PD) \cite{schmidt2010syndrome,kampf2012decoding,nielsen2015sub}.
For both RS and $1$-H codes, PD is only able to correct as many errors as the Sudan algorithm, which is a special case of the GS algorithm.
Recently \cite{nielsen2016power}, PD for RS codes was improved to correct as many errors as the GS algorithm.

In this paper, we combine the idea of improved power decoding (IPD) for RS codes from \cite{nielsen2016power} with the description of PD for $1$-H codes from \cite{nielsen2015sub} in order to obtain an IPD algorithm for $1$-H codes.
Similar to the RS case, we derive a larger system of non-linear key equations (cf.~Section~\ref{sec:key_equations}) than in classical PD and reduce the decoding problem to a linear Pad\'e approximation problem whose solution is likely to agree with the solution of the system of key equations (cf.~Section~\ref{sec:pade_approximation}).
Using a linear-algebraic argument, we derive an upper bound on the maximum number of errors which can possibly be corrected by the decoder (cf.~Section~\ref{sec:decoding_radius}).
In Section~\ref{sec:complexity}, we show that the algorithm can be implemented with sub-quadratic complexity in the code length $n$.
Finally, we present simulation results for various code and decoder parameters which indicate that the new IPD algorithm has a similar failure probability as the GS algorithm for the same parameters and decoding radius (cf.~Section~\ref{sec:numerical_results}).

Besides the theoretical interest in having different decoding paradigms, we see two advantages of the new decoder:
Firstly, the algorithm does not require a root-finding step, which is often considered to be computationally heavy, especially in practical implementations, see e.g.~\cite{ahmed2011vlsi}.
Secondly, the IPD algorithm for RS codes \cite{nielsen2016power} was recently generalized to interleaved RS codes \cite{puchinger2017decoding}, where it improves upon existing decoding algorithms at all rates, including those methods which are based on the GS decoder. It is reasonable to assume that a similar generalization is also possible for $1$-H codes.

\section{Preliminaries}

Let $q$ be a prime power. We follow the notation of \cite{nielsen2015sub}.
The \emph{Hermitian curve} $\H/\Fqtwo$ is the smooth projective plane curve defined by the affine equation $Y^q+Y=X^{q+1}$.
The curve $\H(\Fqtwo)$ has genus $g = \tfrac{1}{2}q(q-1)$ and $q^3+1$ many $\Fqtwo$-rational points $\Pset = \{P_1,\dots,P_{q^3},\Pinf\}$, where $\Pinf$ denotes the point at infinity.
We define $\Rback := \cup_{m\geq0} \L(m\Pinf) = \Fqtwo[X,Y]/(Y^q+Y-X^{q+1})$, which has an $\Fqtwo$-basis of the form $\{ X^iY^j : 0 \leq i, 0 \leq j < q\}$. The order function $\degH : \Rback \to \NN_0 \cup \{-\infty\}, f \mapsto -v_\Pinf(f)$ is defined by the valuation $v_\Pinf$ at $\Pinf$. As a result, we have $\degH(X^iY^j) = i q + j(q+1)$.

Let $n=q^3$ and $m \in \NN$ with $2(g-1) < m <n$. The \emph{one-point Hermitian code} of length $n$ and parameter $m$ over $\Fqtwo$ is defined by
\begin{align*}
\CHerm = \left\{ \left( f(P_1) , \dots , f(P_n) \right) : f \in \L(m\Pinf) \right\}.
\end{align*}
The dimension of $\CHerm$ is given by $k=m-g+1$ and the minimum distance $d$ is lower-bounded by the \emph{designed minimum distance} $\ddesigned := n-m$.

\section{System of Key Equations}
\label{sec:key_equations}

In this section, we derive the system of key equations that we need for decoding, using the same trick as \cite{nielsen2016power} for Reed--Solomon codes.
We use the description of power decoding for one-point Hermitian codes as in \cite{nielsen2015sub}.
Suppose that the received word is $\r = \c + \e \in \Fqtwo^n$, consisting of an error $\e = (e_1,\dots,e_n)$ and a codeword $\c \in \CHerm$, which is obtained from the \emph{message polynomial} $f \in \L(m\Pinf)$. We denote the set of \emph{error positions} by $\Eset = \{i : e_i \neq 0\}$.

In the following sections we show how to retrieve the message polynomial $f$ from the received word $\r$ if the \emph{number of errors}, the Hamming weight $\wtH(\e) = |\Eset|$ of the error, does not exceed a certain decoding radius $\tau$, which depends on the parameters of the decoding algorithm.

A non-zero polynomial $\Lambda \in \L\left(-\sum_{i \in \Eset} P_i + \infty \Pinf\right)$ is called \emph{error locator}. It is well-known that there is an error locator of degree $|\Eset| \leq \degH \Lambda \leq |\Eset|+g$, cf.~\cite{nielsen2015sub}.
and that any error locator fulfills $\degH(\Lambda) \geq |\Eset|$ (cf.~\cite{nielsen2015sub}).
In this section, let $\Lambda$ be some error locator.

\begin{lemma}[\!\!{\cite[Lemma~6]{nielsen2015sub}}]
\label{lem:IP}
There is a polynomial $R \in \Rback$ with $\degH(R)<n+2g$ that satisfies $R(P_i) = r_i$ for all $P_i \in \Pstar$.
\end{lemma}

In the following, let $R \in \Rback$ be as in Lemma~\ref{lem:IP} and $G \in \Rback$ be defined as
\begin{align*}
G = \textstyle\prod_{\alpha \in \Fqtwo} (X-\alpha) = X^{q^2}-X.
\end{align*}
By \cite[Theorem~24]{nielsen2015sub}, we know that there is a unique \emph{error evaluator polynomial} $\Omega \in \Rback$ that fulfills $\Lambda(R-f) = \Omega G$.

The following theorem states the system of key equations that we will use for decoding in the next sections.
Note that the formulation is similar to its Reed--Solomon analog \cite[Theorem~3.1]{nielsen2016power}, with the difference that all involved polynomials are elements of the ring $\Rback$.

\begin{theorem}[System of Key Equations]\label{thm:key_equations}
Let $f$, $\Lambda$, $G$, $R$, and $\Omega$ be as above, and $\ell,s \in \NN$ such that $s \leq \ell$. Then, as a congruence over $\Rback$,
\begin{align}
\Lambda^s f^t &= \sum\limits_{i=0}^{t} \Lambda^{s-i} \Omega^i\tbinom{t}{i} R^{t-i} G^i &&\forall t=1,\dots,s-1, \label{eq:key_eq_equality}\\
\Lambda^s f^t &\equiv \sum\limits_{i=0}^{s-1} \Lambda^{s-i} \Omega^i\tbinom{t}{i} R^{t-i} G^i \mod G^s &&\forall t=s,\dots,\ell. \label{eq:key_eq_congruence}
\end{align}
\end{theorem}

\begin{proof}
We know that $\Omega G = \Lambda\left(f-R\right)$. Thus, for $s,t \in \NN$, we have
\begin{align*}
\Lambda^s f^t = \Lambda^s \left(R + \left(f-R\right)\right)^t
= \sum\limits_{i=0}^{t} \tbinom{t}{i} \Lambda^s \left(f-R\right)^i R^{t-i}.
\end{align*}
In all summands with $i<s$, we can rewrite $\Lambda^s \left(f-R\right)^i = \Lambda^{s-i} (\Lambda\left(f-R\right))^i = \Lambda^{s-i} (\Omega G)^i$. If $i \geq s$, $\Lambda^s \left(f-R\right)^i = (\Lambda(f-R))^s (f-R)^i = (\Omega G)^s (f-R)^{s-i}$, so all those summands are divisible by $G^s$, resulting in
\begin{align*}
\Lambda^s f^t = \sum\limits_{i=0}^{\min\{t,s-1\}} \Lambda^{s-i} \Omega^i\binom{t}{i} R^{t-i} G^i + G^s \left(\sum\limits_{i=\min\{t+1,s\}}^{t} \binom{t}{i} \Omega^s (f-R)^{s-i} R^{t-i} \right).
\end{align*}
For $t < s$, we obtain \eqref{eq:key_eq_equality} since the second part of the sum vanishes and for $t \geq s$, \eqref{eq:key_eq_congruence} holds because the latter sum is divisible by $G^s$. \qed
\end{proof}

\section{Solving the System of Key Equations}
\label{sec:pade_approximation}

The idea of decoding is to find the message polynomial $f$ from the known polynomials $R$ and $G$.
Since the system of key equations is non-linear in the unknown polynomials $\Lambda$, $\Omega$, and $f$, we cannot solve it directly.
Instead, we consider the following linearized problem, a \emph{Pad\'e approximation} problem.

\begin{problem}\label{prob:Pade_R}
Let $G$ and $R$ be as in Section~\ref{sec:key_equations}. Given $\tau \in \NN$ and
\begin{align*}
\asummands{t,i} := \binom{t}{i} R^{t-i} G^i, \quad
G_t := \begin{cases}
x^{\lfloor \frac{t(n+2g-1)+\tau}{q}\rfloor+1}, &\forall t=1,\dots,s-1, \\
G^s, &\forall t=s,\dots,\ell,
\end{cases}
\end{align*}
for all $i=0,\dots,s-1$ and $t=1,\dots,\ell$, find a vector
\begin{align*}
(\lambdai{0},\dots,\lambdai{s-1},\psii{1},\dots,\psii{\ell}) \in \Rback^{s+\ell} \setminus \{ \Zmatrix \},
\end{align*}
with minimal $\degH(\lambdai{0})$ which satisfies
\begin{align}
\sum\limits_{i=0}^{s-1} \lambdai{i} \asummands{t,i} &\equiv \psii{t} \mod G_t  \quad &&\forall t=1,\dots,\ell, \label{eq:Pade_R_congruence}\\
\degH(\lambdai{i}) &\leq s\tau + i(2g-1) \quad &&\forall i = 0,\dots,s-1, \label{eq:Pade_R_degree_1} \\
\degH(\psii{t}) &\leq s\tau + tm \quad &&\forall t = 1,\dots,\ell, \label{eq:Pade_R_degree_2}
\end{align}
where the congruences are over $\Rback$.
\end{problem}

The following theorem motivates the statement of Problem~\ref{prob:Pade_R} by showing that the polynomials $\Lambda^{s-i}\Omega^i$ and $\Lambda^sf^t$ that occur in the key equation fulfill the congruences and degree constraints of the problem.
The minimality condition ensures that if the problem solution corresponds to an error locator $\Lambda$ of some error vector $\e$ (not necessarily the same $\e$ as in Section~\ref{sec:key_equations}), i.e., $\lambdai{0} = \Lambda^s$, then it is the one of smallest degree, and thus hopefully the one corresponding to the $\e$ of smallest Hamming weight.

The theorem also implies a strategy to obtain $f$ after having solved Problem~\ref{prob:Pade_R}:
If the solution of Problem~\ref{prob:Pade_R} results in $\lambdai{0} = \Lambda^s$ and $\psii{1} = \Lambda^sf$ for some error locator $\Lambda$, we divide $\psii{1}$ by $\lambdai{0}$. See~\cite{nielsen2015sub} for how this division can be performed.

\begin{theorem}\label{thm:decoding_solution_is_Pade_solution}
Let $f$, $\Lambda$, $G$, $R$, and $\Omega$ be as in Section~\ref{sec:key_equations}, and $\ell,s \in \NN$ such that $s \leq \ell$.
For $t=1,\dots,\ell$ and $i=0,\dots,s-1$, we define the polynomials (all in $\Rback$)
\begin{align*}
\Lambdai{i} := \Lambda^{s-i}\Omega^i, \quad
\Psii{t} := \Lambda^s f^t,
\end{align*}
Then, $(\Lambdai{0},\dots,\Lambdai{s-1},\Psii{1},\dots,\Psii{\ell})$ satisfies Conditions \eqref{eq:Pade_R_congruence} - \eqref{eq:Pade_R_degree_2} of Problem~\ref{prob:Pade_R} for any $\tau \geq \degH(\Lambda)$.
\end{theorem}

\begin{proof}
Inequality \eqref{eq:Pade_R_degree_1} is fulfilled since
\begin{align*}
\degH(\Omega) = \degH(\Omega G) - n = \degH(\Lambda(f-R)) -n \leq \tau + 2g-1.
\end{align*}
Also, Inequality~\eqref{eq:Pade_R_degree_2} holds due to $\degH(f) \leq m$ and
\begin{align*}
\degH(\Psii{t}) = \degH(\Lambda^s) + \degH(f^t) \leq s\tau+tm.
\end{align*}
Condition \eqref{eq:Pade_R_congruence} is satisfied by Theorem~\ref{thm:key_equations} (note that $\asummands{t,i} = 0$ for $i>t$ and that the congruence modulo $x^{\lfloor \frac{t(n+2g-1)+\tau}{q}\rfloor+1}$ in \eqref{eq:Pade_R_congruence} is the same as equality due to the degree restrictions). \qed
\end{proof}

\section{Decoding Radius and Failure Behavior}
\label{sec:decoding_radius}

As any other power decoder, the new decoding algorithm is a partial decoding algorithm, which means that it might fail for certain error patterns.
This failure behavior has many reasons that we would like to discuss in this section.
We start by deriving a bound on the parameter $\tau$ of Problem~\ref{prob:Pade_R} that ensures the problem to have a solution.
\begin{theorem}\label{thm:decoding_radius}
Problem~\ref{prob:Pade_R} is guaranteed to have a solution if
\begin{align*}
\tau \geq \taunew := n \left[1-\tfrac{s+1}{2(\ell+1)}\right] -\tfrac{\ell}{2s} m - \tfrac{\ell-s+1}{s(\ell+1)} + \tfrac{g-1}{\ell+1}.
\end{align*}
\end{theorem}

\begin{proof}
Problem~\ref{prob:Pade_R} is guaranteed to have a solution if there at least one vector
\begin{align*}
(\lambdai{0},\dots,\lambdai{s-1},\psii{1},\dots,\psii{\ell}) \in \Rback^{s+\ell} \setminus \{ \Zmatrix \},
\end{align*}
satisfying Conditions \eqref{eq:Pade_R_congruence}, \eqref{eq:Pade_R_degree_1}, and \eqref{eq:Pade_R_degree_2}.
We can find such a solution by solving the following homogeneous linear system of equations in the coefficients of the $\lambdai{i}$, which we consider these coefficients as indeterminates.
Since $\degH( \psii{t}) \leq s\tau+tm$ (cf.~\eqref{eq:Pade_R_degree_2}), the coefficients of $\sum_{i=0}^{s-1} \lambdai{i} \asummands{t,i}$ in \eqref{eq:Pade_R_congruence} of degree greater than $s\tau+tm$ and less than
\begin{align*}
T_t := \begin{cases}
t(n+2g-1) + s\tau + 1, &t=1,\dots,s-1, \\
\degH(G^s), &t=s,\dots,\ell,
\end{cases}
\end{align*}
must be zero.
Since we require $\degH \lambdai{i} \leq s\tau+i(2g-1)$, see \eqref{eq:Pade_R_degree_1}, there at at least $s\tau+i(2g-1)-g+1$ indeterminates for $\lambdai{i}$.
After obtaining non-zero polynomials $\lambdai{i}$, we can find $\psii{t}$ by computing $\sum_{i=0}^{s-1} \lambdai{i} \asummands{t,i}$ modulo $G_t$.

It suffices to show that the described system has a non-zero solution for $\tau \geq \taunew$.
The system has at most
\begin{align*}
\NumEq &= \textstyle\sum_{t=1}^{\ell} \left[ T_t - (s\tau+tm) - 1 \right] \\
&\leq n s\left[ \ell+1-\tfrac{s+1}{2} \right] - \tfrac{\ell(\ell+1)}{2} m + \tfrac{s(s-1)}{2} (2g-1) + \tau s(s-1-\ell) - (\ell-s+1)
\end{align*}
equations and at least
\begin{align*}
\NumVar &= \textstyle\sum_{i=0}^{s-1} \left[ s\tau+i(2g-1) - g + 1 \right] = s^2\tau + \tfrac{s(s-1)}{2} (2g-1) - sg + s
\end{align*}
indeterminates. Thus, it has a non-zero solution if $\NumVar \geq 1+ \NumEq$, which can be re-written as $\tau \geq \taunew$.
\qed
\end{proof}

Theorem~\ref{thm:decoding_radius} can be interpreted as follows.
For some $\tau \in \NN$, we denote by $\SolSpace_\tau$ the $\Fqtwo$-vector space consisting of all vectors 
\begin{align*}
(\lambdai{0},\dots,\lambdai{s-1},\psii{1},\dots,\psii{\ell}) \in \Rback^{s+\ell}
\end{align*}
that satisfy the congruences and degree constraints of Problem~\ref{prob:Pade_R} with parameter $\tau$.
If we choose $\tau \geq \taunew$, then $\dim_{\Fqtwo}(\SolSpace_\tau) \geq 1$.
In addition, if $\tau \geq \degH(\Lambda)$, then
\begin{align*}
(\Lambda^s, \Lambda^{s-1}\Omega,\dots,\Lambda\Omega^{s-1},\Lambda^s f,\Lambda^s f^2,\dots,\Lambda^s f^\ell) \in \SolSpace_\tau.
\end{align*}
Hence, if there is a $\tau$ with $|\Eset| \leq \deg \Lambda \leq \tau \leq \taunew$ and $\dim_{\Fqtwo} (\SolSpace_\tau) = 1$, a non-trivial solution of Problem~\ref{prob:Pade_R} must yield a solution $(\Lambda^s, \Lambda^s f)$ of the decoding problem.
Thus, we could expect that at least in some cases, we can decode up to $|\Eset| \leq \taunew$ errors. However, there are several problems that could prevent us from correcting $\taunew$ many errors:
\begin{enumerate}[label=\roman*)]
\item The minimal degree of an error locator is greater than $|\Eset|$. Recall that it is only guaranteed that there is an error locator of $\degH \Lambda \leq |\Eset|+g$.
\item We get $\dim_{\Fqtwo} (\SolSpace_\tau) > 1$ already for some $\tau<\taunew$. This can have two reasons:
\begin{itemize}
\item The number of equations is smaller than $E$ (as in the proof of Theorem~\ref{thm:decoding_radius}), which can be the case if $\degH R < n+2g-1$.
\item The equations are linearly dependent.
\end{itemize}
\item There is no $\tau$ with $\dim_{\Fqtwo} (\SolSpace_\tau) = 1$ (e.g., if there is a $\tau$ with $\dim_{\Fqtwo} (\SolSpace_\tau) = 0$ and $\dim_{\Fqtwo} (\SolSpace_{\tau+1}) > 1$) and there is a ``smaller'' solution (corresponding to another codeword or a generic one) in $\SolSpace_\tau$ than $(\Lambda^s,\dots,\Lambda\Omega^{s-1},\Lambda^s f,\dots,\Lambda^s f^\ell)$.
\end{enumerate}

We will see in the Section~\ref{sec:numerical_results} that in our experiments, for all tested examples, we were able to correct up to $n [1-\tfrac{s+1}{2(\ell+1)}] -\tfrac{\ell}{2s} m - \tfrac{\ell-s+1}{s(\ell+1)} = \taunew- \tfrac{g-1}{\ell+1}$ many errors with high probability. This number of errors coincides with the classical power decoding radius for $s=1$, cf.~\cite{nielsen2015sub}.

\section{Complexity}
\label{sec:complexity}

In this section, we show that Problem~\ref{prob:Pade_R} can be solved in sub-quadratic time in the code length $n$. We use the algorithm in~\cite{rosenkilde2017SimHermPade}, which computes, for given $S_{\i,\j} \in \Fqtwo$, $G_\j \in \Fqtwo[X]$, $T_\i \in \NN$, and $N_\i \in \NN$, where $\i \in I$ and $\j \in J$ (index sets), a basis (the solution space is a vector space) of all solutions $\lambda_\i,\psi_\j \in \Fqtwo[X]$, that fulfill
\begin{align*}
\sum_{\i \in I} \lambda_\i &\equiv \psi_\j \mod G_\j &&\forall \j \in J, \\
\deg \lambda_\i &\leq N_\i &&\forall \i \in I, \\
\deg \psi_\i &\leq T_\i &&\forall \j \in J,
\end{align*}
in $O^\sim\left(|J|^{\omega-1} \cdot |I| \cdot \max_\j\{\deg G_\j\}\right)$ operations over $\Fqtwo$, where $\omega$ is the matrix multiplication exponent.

We use the $\FX$-vector representation of an element of $\Rback$ (cf.~\cite{nielsen2015sub}) to reformulate Problem~\ref{prob:Pade_R} as a problem of the type above. Recall that for $a \in \Rback$, we can write $a = \sum_{i=0}^{q-1} a_i Y^i \in \Rback$ with unique $a_i \in \FX$. Then, the \emph{vector representation} \cite{nielsen2015sub} of $a$ is defined by $\vr(a) = (a_0,\dots,a_{q-1}) \in \FX^q$.
Note that $q\deg(a_i)+i(q+1)\leq \degH(a)$. For $a,b \in \Rback$ it can be shown that
\begin{align*}
\vr(a+b) = \vr(a) + \vr(b), \qquad
\vr(ab) = \vr(a) \mr(b) \XiM, 
\end{align*}
where $\mr(b) \in \FX^{q \times (2q-1)}$ and $\XiM \in \FX^{(2q-1) \times q}$ are defined by
\begin{align*} \small
\mr(b) := 
\begin{bmatrix}
b_0 & b_1 & b_2 & \dots & b_{q-1} & & &\\
    & b_0 & b_1 & \dots & b_{q-2} & b_{q-1} & &\\
    &     & \ddots & \ddots & \dots & \ddots & \ddots & \\
    &     &        & b_0    & b_1 & \dots & b_{q-2} & b_{q-1} 
\end{bmatrix},
\quad
\XiM
:=
\begin{bmatrix}
1 &   &        &    \\
  & 1 &        &    \\
  &   & \ddots &    \\
  &   &        & 1  \\
X^{q+1}  & -1  &   &    \\
 & X^{q+1}  & -1  &    \\
 & \ddots  & \ddots  &    \\
 &   & X^{q+1}   &  -1  \\
\end{bmatrix}.
\end{align*}
Using this notation, we define $\AM{t,i} := \mr(\asummands{t,i}) \XiM \in \FX^{q \times q}$. We are ready to state the final complexity result.

\begin{theorem}\label{thm:complexity}
Problem~\ref{prob:Pade_R} can be solved using the algorithm in \cite{rosenkilde2017SimHermPade} with
\begin{align*}
&I = \left\{(i,j): i \in \{0,\dots,s-1\}, \, j \in \{0,\dots,q-1\} \right\}, &S_{(i,j),(t,r)}& = \AMe{t,i}_{j,r},\\
&J = \left\{ (t,r) : t \in \{1,\dots,\ell\}, \, r \in \{0,\dots,q-1\} \right\}, &G_{(t,r)}& = G_t, \\
&N_{(i,j)} = \tfrac{s \tau+i(2g-1)-j(q+1)}{q}, \quad T_{(t,r)} = \tfrac{s \tau+tm-r(q+1)}{q} 
\end{align*}
in $O^\sim( \ell^{\omega-1} s^2 n^{\frac{\omega+2}{3}})$ operations over $\Fq$, where the $O^\sim$ hides $\log(ns\ell)$ factors.
\end{theorem}

\begin{proof}
Similar to \cite{nielsen2015sub}, pre-computing the matrices $\AM{t,i}$ is negligible compared to solving the Pad\'e approximation problem.
By the properties of $\vr(\cdot)$, it is clear that $\lambdai{i},\psii{t} \in \Rback$ solve Problem~\ref{prob:Pade_R} if and only if $(\lambda_{(i,0)},\dots,\lambda_{(i,q-1)}) = \vr(\lambdai{i})$ and $(\psi_{(t,0)},\dots,\psi_{(t,q-1)}) = \vr(\psii{t})$ correspond to a non-zero element in the output of the algorithm in \cite{rosenkilde2017SimHermPade} of minimal $\max_{j \in \{0,\dots,q-1\}}\left\{q \deg(\lambda_{(0,j)}) + (q+1)j \right\}$. Since $\deg G_t \leq \lfloor \frac{t(n+2g-1)+\tau}{q}\rfloor+1 \in O(sn/q)$, a basis of the solution space is found in
\begin{align*}
O^\sim\left( (\ell q)^{\omega-1} (sq) (sn) \right) = O^\sim\left( \ell^{\omega-1} s^2 q^{\omega-1} n \right) = O^\sim\left( \ell^{\omega-1} s^2 n^{\frac{\omega+2}{3}} \right).
\end{align*}
The algorithm in \cite{rosenkilde2017SimHermPade} outputs a reduced basis, so a minimal element is guaranteed to be one of the basis elements. \qed
\end{proof}

Note that for constant parameters $\ell,s$, the complexity in Theorem~\ref{thm:complexity} is sub-quadratic in the code length $n$. We achieve the same complexity\footnote{The exponent of $\ell$ in the complexity statements in \cite{nielsen2015sub} is $\omega$. If we apply the algorithm from~\cite{rosenkilde2017SimHermPade} to these methods, we will also get $\omega-1$.} as the algorithms in \cite{nielsen2015sub}.

\section{Numerical Results}
\label{sec:numerical_results}

In this section, we present simulation results.
We have conducted Monte-Carlo simulations for estimating the failure probability of the new improved power ($\PfailPower$) and the Guruswami--Sudan ($\PfailGS$) decoder in a channel that randomly adds $\tau$ errors, using a sample size $N \in \{10^3,10^4\}$.
The decoder was implemented in SageMath~v7.5 \cite{stein_sagemath_????}, based on the power decoder implementation of \cite{nielsen2015sub}. We used the Guruswami--Sudan decoder implementation from \cite{nielsen2015sub}.

\begin{table}[h!]
\caption{Observed failure rate of the improved power ($\PfailPower$) and Guruswami--Sudan ($\PfailGS$) decoder. Code parameters $q,m,n,k,\ddesigned$. Decoder parameters $\ell,s$. Number of errors $\tau$ ($^\ast$decoding radius as in \eqref{eq:decoding_radius_probably_real}). Number of experiments $N$.}
\label{tab:failure_rate}
\centering
{
\renewcommand{\arraystretch}{1.2}
\setlength{\tabcolsep}{5pt}
\begin{tabular}{c|c||c|c|c||c|c||c|c|c||c}
$q$ & $m$ & $n$ & $k$ & $\ddesigned$ & $\ell$ & $s$ & $\tau$ & $\PfailPower$ & $\PfailGS$  & $N$ \\
\hline \hline
$4$ & $15$ &  $64$ & $10$ &  $49$ & $4$ & $2$ & $28$ & $0$ & $0$ & $10^4$ \\
    &      &       &      &       &     &     & $\phantom{{}^\ast}29^\ast$ & $0$ & $3.30 \cdot 10^{-3}$ & $10^4$ \\
    &      &       &      &       &     &     & $30$ & $9.93 \cdot 10^{-1}$ & $9.39 \cdot 10^{-1}$ & $10^4$ \\
\hline
$5$ & $55$ & $125$ & $46$ &  $70$  & $3$ & $2$ & $34$ & $0$ & $0$ & $10^4$ \\
    &      &       &      &      &     &     & $35$ & $0$ & $0$ & $10^4$ \\
    &      &       &      &      &     &     & $\phantom{{}^\ast}36^\ast$ & $0$ & $4.00 \cdot 10^{-4}$ & $10^4$ \\
\hline
$5$ & $20$ & $125$ & $11$ & $105$ & $5$ & $2$ & $67$ & $0$ & $0$ & $10^3$ \\
    &      &       &      &       &     &     & $\phantom{{}^\ast}68^\ast$ & $0$ & $7.00 \cdot 10^{-3}$ & $10^3$ \\
    &      &       &      &       &     &     & $69$ & $9.91 \cdot 10^{-1}$ & $9.60 \cdot 10^{-1}$ & $10^3$ \\
\hline
$7$ & $70$ & $343$ & $50$ & $273$ & $3$ & $2$ & $160$ & $0$ & $0$ & $10^3$ \\
    &      &       &      &       &     &     & $\phantom{{}^\ast}161^\ast$ & $0$ & $0$ & $10^3$ \\
    &      &       &      &       &     &     & $162$ & $9.78 \cdot 10^{-1}$ & $9.86 \cdot 10^{-1}$ & $10^3$ \\
\hline
$7$ & $70$ & $343$ & $50$ & $273$ & $4$ & $2$ & $168$ & $0$ & $0$ & $10^3$ \\
    &      &       &      &       &     &     & $\phantom{{}^\ast}169^\ast$ & $0$ & $0$ & $10^3$ \\
    &      &       &      &       &     &     & $170$ & $9.79 \cdot 10^{-1}$ & $2.2 \cdot 10^{-2}$ & $10^3$ \\
    &      &       &      &       &     &     & $171$ & $1$ & $1$ & $10^3$ \\
\hline 
$7$ & $55$ & $343$ & $35$ & $288$ & $4$ & $2$ & $\phantom{{}^\ast}184^\ast$ & $0$ & $0$ & $10^3$ \\
    &      &       &      &       &     &     & $185$ & $9.82 \cdot 10^{-1}$ & $1.9 \cdot 10^{-2}$ & $10^3$ \\
    &      &       &      &       &     &     & $186$ & $1$ & $1$ & $10^3$
\end{tabular}
}
\end{table}

Table~\ref{tab:failure_rate} presents the simulation results for various code ($q,m,n,k,\ddesigned$), decoder ($\ell,s$), and channel ($\tau$) parameters.
It can be observed that both algorithms can almost always correct
\begin{align}
\tau = n \left[1-\tfrac{s+1}{2(\ell+1)}\right] - \tfrac{\ell}{2s} m - \tfrac{\ell-s+1}{s(\ell+1)} \label{eq:decoding_radius_probably_real}
\end{align}
errors, improving upon classical power decoding. Also, none of the two algorithms is generally superior.

\section{Conclusion}

We have presented a new decoding algorithm for one-point Hermitian codes which is based on the improved power decoder for Reed--Solomon codes from~\cite{nielsen2016power}.
Experimental results indicate that the new algorithm has a similar failure probability as the Guruswami--Sudan algorithm at the same decoding radius.

A generalization of the new algorithm to interleaved one-point Hermitian codes, similar to \cite{puchinger2017decoding}, promises improved decoding radii for interleaving degrees $m>1$ compared to existing decoding algorithms, and is work in progress.

\bibliographystyle{IEEEtran}
\bibliography{main}

\end{document}

%% file: defs.tex
\usepackage[utf8]{inputenc}
\usepackage[english]{babel}
\usepackage[T1]{fontenc}
\usepackage{IEEEtrantools}
\usepackage[colorlinks=true,citecolor=black,linkcolor=black,urlcolor=black]{hyperref}
\usepackage{epsfig}
\usepackage{amsmath, amssymb, amsbsy}
\usepackage{xspace} 
\usepackage{color}
\usepackage{url}             
\usepackage{cite}            
\usepackage{booktabs} 
\usepackage{parskip} 
\usepackage{enumitem}
\usepackage{tikz}
\usepackage[linesnumbered,ruled,vlined,titlenumbered]{algorithm2e}
\usepackage{mathtools}
\usepackage{xparse}
\usepackage{url}

\def\ve#1{{\mathchoice{\mbox{\boldmath$\displaystyle #1$}}%
              {\mbox{\boldmath$\textstyle #1$}}%
              {\mbox{\boldmath$\scriptstyle #1$}}%
              {\mbox{\boldmath$\scriptscriptstyle #1$}}}}

\renewcommand{\H}{\mathcal{H}}

\newcommand{\NN}{\mathbb{N}}

\renewcommand{\L}{\mathcal{L}}

\newcommand{\Fq}{\mathbb{F}_{q}}

\renewcommand{\H}{\mathcal{H}}

\newcommand{\Fqtwo}{\mathbb{F}_{q^2}}

\newcommand{\Pinf}{{P_\infty}}
\newcommand{\Pset}{\mathcal{P}}
\newcommand{\Pstar}{\mathcal{P}^\ast}

\renewcommand{\r}{\ve r}
\renewcommand{\c}{\ve c}
\newcommand{\e}{\ve e}

\newcommand{\CHerm}{\mathcal{C}_\mathcal{H}}
\newcommand{\Eset}{\mathcal{E}}
\newcommand{\degH}{\deg_{\mathcal{H}}}

\newcommand{\Rback}{\mathcal{R}}

\newcommand{\vr}{{\ve \nu}}
\newcommand{\mr}{{\ve \mu}}
\newcommand{\XiM}{{\ve \Xi}}

\newcommand{\FX}{\Fqtwo[X]}

\newcommand{\asummands}[1]{a^{(#1)}}
\newcommand{\AM}[1]{\Amatrix^{(#1)}}
\newcommand{\AMe}[1]{A^{(#1)}}
\newcommand{\Lambdai}[1]{\Lambda^{(#1)}}
\newcommand{\lambdai}[1]{\lambda^{(#1)}}

\newcommand{\Psii}[1]{\Psi^{(#1)}}
\newcommand{\psii}[1]{\psi^{(#1)}}

\newcommand{\Zmatrix}{\ve 0}
\newcommand{\Amatrix}{\ve A}

\newcommand{\ddesigned}{d^\ast}

\newcommand{\wtH}{\mathrm{wt}_\mathrm{H}}

\newcommand{\taunew}{\tau_\mathrm{new}}
\newcommand{\PfailPower}{\mathrm{\hat{P}}_\mathrm{fail, IPD}}
\newcommand{\PfailGS}{\mathrm{\hat{P}}_\mathrm{fail, GS}}

\newcommand{\NumEq}{E}
\newcommand{\NumVar}{V}

\renewcommand{\i}{\ve{i}}
\renewcommand{\j}{\ve{j}}

\newcommand{\SolSpace}{\mathcal{V}}